\documentclass[referee,pdflatex,sn-basic,Numbered]{sn-jnl}

\usepackage{amsmath,amssymb,mathtools}
\usepackage{booktabs,array}
\usepackage{xcolor}
\usepackage{microtype}
\hypersetup{
  hidelinks,
  pdftitle={Key Recovery from Residue-Confined Errors in the Pradhan CRT-RLWE Construction},
  pdfauthor={Lukasz Olejnik and Bartosz Naskrecki},
  pdfsubject={Error collapse and error-distribution mismatch in the Pradhan CRT-RLWE construction},
  pdfkeywords={CRT-RLWE, Ring-LWE, homomorphic encryption, error distributions, CRT encoding}
}

\theoremstyle{thmstyleone}
\newtheorem{theorem}{Theorem}[section]
\newtheorem{proposition}[theorem]{Proposition}
\newtheorem{lemma}[theorem]{Lemma}
\newtheorem{corollary}[theorem]{Corollary}
\theoremstyle{thmstyletwo}
\newtheorem{remark}[theorem]{Remark}
\newtheorem{example}[theorem]{Example}
\theoremstyle{thmstylethree}
\newtheorem{definition}[theorem]{Definition}

\newcommand{\editionlabel}{}
\newcommand{\claimid}[1]{}
\newcommand{\provenance}[1]{}
\newcommand{\internalnote}[1]{}

\newcommand{\Z}{\mathbb Z}
\newcommand{\F}{\mathbb F}
\newcommand{\R}{\mathcal R}
\newcommand{\Rq}{\R_q}
\newcommand{\units}[1]{#1^{\times}}

\newcommand{\Adv}{\operatorname{Adv}}
\newcommand{\CRT}{\operatorname{CRT}}
\newcommand{\Enc}{\mathsf{Enc}}

\newcommand{\KeyGen}{\mathsf{KeyGen}}
\newcommand{\TV}{\operatorname{TV}}

\begin{document}

\title[Key Recovery from Residue-Confined Errors in Pradhan CRT-RLWE]
{Key Recovery from Residue-Confined Errors in the Pradhan CRT-RLWE
Construction\editionlabel}

\author*[1,2]{\fnm{Lukasz} \sur{Olejnik}}\email{me@lukaszolejnik.com}
\author[3,4]{\fnm{Bartosz} \sur{Naskrecki}}\email{bartnas@amu.edu.pl}

\affil*[1]{\orgdiv{Department of War Studies},
\orgname{King's College London},
\orgaddress{\city{London}, \country{United Kingdom}}}
\affil[2]{\orgname{Independent Researcher}}
\affil[3]{\orgdiv{Faculty of Mathematics and Computer Science},
\orgname{Adam Mickiewicz University},
\orgaddress{\city{Poznan}, \country{Poland}}}
\affil[4]{\orgdiv{Centre for Credible AI},
\orgname{Warsaw University of Technology},
\orgaddress{\city{Warsaw}, \country{Poland}}}

\abstract{We show that the CRT-FHE scheme of Pradhan et al.\ is insecure for laws
within its assumed error distribution range.  The secret key follows from the public key by a single ring inversion whenever the public multiplier is a unit.  The
plaintext is recovered from any ciphertext under such a law without the secret
key, for every multiplier, giving chosen-plaintext advantage $1/2$.  We further show that the
transformation from ordinary Ring-LWE to CRT-RLWE does not preserve the error
distribution, so it does not establish that CRT-RLWE is at least as hard as
Ring-LWE.

One mechanism underlies both.  The Chinese remainder theorem (CRT) function is
reduced modulo $p_1p_2$ while its output is used modulo a coprime modulus $q$,
so under every zero-preserving section an error in $p_2\R$ encodes to zero.  The law $p_2B_1$ is so confined, meets the
stated conditions, and decrypts correctly.  Confinement is not a
weakness of scale: scaling any baseline law by $p_2$ leaves its ordinary
Ring-LWE problem exactly equivalent, while the reduced encoder destroys every
error it produces.  The reduction discrepancy is a multiple of
$p_1p_2$ and not of $q$, so the small-error premise of the proof cannot remove
it, and at the reported parameters a single error coefficient refutes the
identity while satisfying that premise.  The centered binomial $B_2$ separates the coefficient
laws at total variation distance $3/8$, and at the reported dimension that distance
between the induced polynomial laws is exponentially close to one.  }

\keywords{Chinese remainder theorem, ring learning with errors,
homomorphic encryption, error distributions, CRT encoding, cryptanalysis}

\maketitle

\section{Introduction}

Ring learning with errors (RLWE) is a family of problems whose security
depends on the ring, modulus, representation, and error distribution
\cite{LPR2010,Peikert2016}.  Hardness theorems apply to specified
distributions. They do not imply that every narrow or algebraically
structured error law is secure.  A transformation of the error must therefore
be accounted for in the definition of the error distribution and in the
security reduction that accompanies it.

Pradhan, Dutta, Jangir, and Das propose a homomorphic-encryption construction
based on a Chinese remainder theorem (CRT) variant of the RLWE distribution
\cite{PradhanEtAl2026}.  Definition~1 gives the CRT function by a formula
reduced modulo $p_1p_2$, $\KeyGen$ and $\Enc$ then insert its output into a
ring modulo a coprime ciphertext modulus $q$.  This cross-modulus use requires
a representative, because equality modulo $p_1p_2$ need not be equality modulo
$q$.  The representative choice is immaterial for the zero CRT class: every
zero-preserving convention sends it to $0$, and that single fact drives the
results below.  We use the definition and theorem numbering of
\cite{PradhanEtAl2026}.

Our contributions are the following.

\begin{enumerate}
\item Residue confinement collapses the reduced CRT error.  If the public-key
error lies in $p_2\R$, then $b=as$ and the secret is recovered exactly on the
unit event (Theorem~\ref{thm:collapse}).  If the challenge errors are confined
as well, a polynomial gcd and CRT decomposition recovers the encoded challenge
message for \emph{every} public multiplier, with chosen-plaintext advantage
$1/2$ (Theorem~\ref{thm:distinguish}).
\item A concrete efficiently samplable subgaussian law realizes those
hypotheses.  The coefficient law $p_2B_1$, built from the centered
Bernoulli difference $B_1$, is discrete, bounded by $p_2$, supported on
$p_2\R$, and subgaussian with variance proxy $p_2^2/2$
(Proposition~\ref{prop:law}).
\item The same-law RLWE transformation does not preserve the error
distribution.  For the least-nonnegative representative, the discrepancy
between the two maps has the closed form $p\lfloor ke/p_2\rfloor$, so under
the proof's no-wrap premise the identity used in Theorem~4 holds on
$0\le ke<p_2$ (Lemma~\ref{lem:discrepancy}), and at the reported parameters that
is $e\in\{0,1\}$.  Distributionally, scaling induces
$(e\mapsto\Delta_2e)_\#\chi$ while the reduced encoder induces
$\psi=(e\mapsto E_{\sigma,q}(0,e))_\#\chi$, and these differ whenever the
coefficient law charges more integers, pairwise distinct modulo $q$, than the
small CRT factor $p_2$ (Proposition~\ref{prop:mismatch} and
Corollary~\ref{cor:mismatch}).
\end{enumerate}

The results of Section~\ref{sec:attack} concern
residue-confined errors, whereas the distributional results of
Section~\ref{sec:mismatch} apply more generally, including to
centered-binomial and untruncated discrete-Gaussian coefficient laws.
The relation to the conditional security theorem of
\cite{PradhanEtAl2026} is discussed in Section~\ref{sec:conditional}.

Evaluation and subfield attacks exploit RLWE or PLWE instances whose errors
become distinguishable under algebraic projections
\cite{EisentraegerHallgrenLauter2014,EliasLauterOzmanStange2015,
ChenLauterStange2015}.  In that line of work the error law is fixed by the
problem statement, and the attack applies a projection under which an already
narrow law becomes statistically visible.  The mechanism here is different and
prior to any projection: the construction interposes an additional encoder
between the sampled error and the ring element that carries it, and that
encoder is not injective on the sampled error.  For the residue-confined laws
of Section~\ref{sec:attack} the resulting failure is exact, and in both sections it is a property of the encoding step, not
 of the ambient number field.  Peikert studies the sensitivity of Ring-LWE
security to the choice and representation of the error distribution, and
separates the narrow laws that admit attacks from those supported by
worst-case hardness results \cite{Peikert2016}.  Theorem~\ref{thm:scaling} establishes an exact equivalence between the
corresponding ordinary-RLWE distributions.  Castryck,
Iliashenko, and Vercauteren analyze how transformations between error
distributions in ring-based LWE affect the resulting problem
\cite{CastryckIliashenkoVercauteren2016}.  The push-forward computed in
Section~\ref{sec:mismatch} is such a transformation, imposed here by the
construction itself.

\section{CRT encoding and representative lifts}

Let
\[
  \R=\Z[X]/(X^n+1),\qquad \R_t=\R/t\R,
\]
where $n$ is a power of two.  Let $p_1,p_2$ be coprime, set $p=p_1p_2$, and
assume $p<q/2$ and $\gcd(p,q)=1$, as required by the parameter conditions of
\cite{PradhanEtAl2026}.  Denote the class-valued CRT isomorphism by
\[
  \overline{\CRT}:\R_{p_1}\times\R_{p_2}\longrightarrow\R_p.
\]
All lifts below act coefficientwise in the power basis.

Choose integers
\[
 \Delta_1=p_2[p_2^{-1}]_{p_1},\qquad
 \Delta_2=p_1[p_1^{-1}]_{p_2},
\]
where the inverses are represented in the indicated least-nonnegative
ranges.  Thus $\Delta_1$ has residues $(1,0)$ and $\Delta_2$ has residues
$(0,1)$ modulo $(p_1,p_2)$.

Throughout, the error distribution $\chi$ samples integer-polynomial
representatives.  This is the type required to make the encoder well defined:
the CRT function consumes $e\bmod p_2$, while the ciphertext lives modulo the
coprime modulus $q$, so reduction modulo $q$ occurs after the encoding step
and not before it.  Where a coefficient law is needed we write $\nu$ for a
distribution on $\Z$ and take $\chi$ to sample the $n$ coefficients
independently from $\nu$, as the centered binomial and discrete Gaussian
instantiations do.

\begin{definition}[Reduced CRT encoder]
\label{def:reduced}
A coefficientwise section is a public map
$\sigma:\R_p\rightarrow\R$ whose reduction modulo $p$ is the identity.
It is zero-preserving if $\sigma(0)=0$.  The associated encoder is
\[
 E_{\sigma,q}(m,e)=
 \sigma\!\left(\overline{\CRT}(m,e\bmod p_2)\right)\bmod q.
\]
The least-nonnegative and centered coefficient representatives are
zero-preserving sections.
\end{definition}

\begin{definition}[Unreduced integer lift]
\label{def:raw}
For a fixed integer representative $\widetilde m$ of the message, define
\[
 E_{\rm raw,q}(m,e)=\Delta_1\widetilde m+\Delta_2e\pmod q.
\]
Here the sampled integer representative $e$ is used before reduction modulo
$p_2$.
\end{definition}

\begin{lemma}[Carry relation between the two lifts]
For every $m$ and $e$, there is an integer polynomial
$K_\sigma(m,e)\in\R$ such that
\[
\Delta_1\widetilde m+\Delta_2e
=
\sigma\!\left(\overline{\CRT}(m,e \bmod p_2)\right)
+pK_\sigma(m,e).
\]
Consequently, the two encoders differ in $R_q$ by
$pK_\sigma(m,e)$, which need not vanish because $p$ is a unit modulo $q$.
\end{lemma}

\begin{proof}
The two integer polynomials have the same coefficient residues modulo $p_1$
and modulo $p_2$, hence their difference is coefficientwise divisible by
$p=p_1p_2$.  Reduction modulo $q$ gives the last statement.
\end{proof}

The separation is immediate for $m=0$ and $e=p_2r$.  The reduced encoder,
realized by any zero-preserving section, gives
\[
 E_{\sigma,q}(0,p_2r)=\sigma(0)=0,
\]
whereas the unreduced lift gives
\[
 E_{\rm raw,q}(0,p_2r)=\Delta_2p_2r\pmod q,
\]
which is not the zero map under the stated parameter conditions.

Multiples of the CRT modulus are retained in the integer-level correctness
analysis of \cite{PradhanEtAl2026} as well.  In its Section~4.1, where
$p=p_1^rp_2$, Equation~(8) writes the centered decryption value as
\[
 v(c)=p\,k(c)+\CRT_{p_1^r,p_2}(m(c),\eta(c))
 =p\,k(c)+\Delta_1m(c)+\Delta_2\eta(c)
\]
with a coarse integer component $k(c)$, and Equation~(9) retains
$p\,k(c)$ in the global error
$E(c)=\Delta_2\eta(c)+p\,k(c)$.  Their $k(c)$ is distinct from the representative carry $K_\sigma(m,e)$ of Lemma~2.3; in particular, Equation~(13)
specializes fresh encryption to $k(c)=0$.
Equations~(8)--(9) nevertheless show that the integer-level correctness
analysis of \cite{PradhanEtAl2026} retains integral multiples of the
applicable CRT modulus.

Lemma~2.3 therefore establishes three facts used below.  The
reduced CRT encoder and the unreduced integer lift are distinct maps into
$\Rq$.  The results of Section~\ref{sec:attack} concern the reduced encoder of
Definition~\ref{def:reduced}, which is a representative-level realization of
the reduced CRT function of Definition~1 of \cite{PradhanEtAl2026}, used in
$\KeyGen$ through the term
$\CRT_{p_1,p_2}(0,e)$.  The algebraic identity used in Theorem~4, which
identifies that term with the integer $\Delta_2e$, corresponds instead to
Definition~\ref{def:raw}.

\section{Consequences of error collapse}
\label{sec:attack}

Call a distribution on integer polynomials $p_2$-confined if its support is
contained in $p_2\R$.  Under the reduced encoder, consider the public-key and
encryption equations
\begin{align*}
 b&=as+E_{\sigma,q}(0,e_{\rm pk}),\\
 c_0&=bu+E_{\sigma,q}(m,e_0),\\
 c_1&=-au+E_{\sigma,q}(0,e_1)
\end{align*}
in $\Rq$.  Write $M(m)=E_{\sigma,q}(m,0)$ for the encoding of a message
carrying no error.

\begin{theorem}[Error collapse and key recovery]
\label{thm:collapse}
Let $\sigma$ be an efficiently computable zero-preserving section.
If $e_{\rm pk}$ is sampled from a $p_2$-confined law, then
\[
 E_{\sigma,q}(0,e_{\rm pk})=0,
 \qquad b=as .
\]
A passive algorithm recovers $s=a^{-1}b$ whenever $a\in\units{\Rq}$, in time
polynomial in $n$ and $\log q$.
\end{theorem}

\begin{proof}
Confinement gives $e_{\rm pk}\bmod p_2=0$, so
$\overline{\CRT}(0,e_{\rm pk}\bmod p_2)=0$, and zero preservation gives
$\sigma(0)=0$.  The public-key equation reduces to $b=as$.  For the inverse,
note that multiplication by $a$ is a $\Z_q$-linear endomorphism of $\Rq$,
and $a$ is a unit if and only if that endomorphism is invertible.  The inverse
$a^{-1}$ is
then obtained by linear algebra over $\Z_q$ in time polynomial in $n$ and
$\log q$.  When $q$ is prime the extended Euclidean algorithm in
$\F_q[X]/(X^n+1)$ is the standard alternative.
\end{proof}

The unit event is not rare.  When $q$ is prime and
$F=X^n+1=\prod_i f_i^{k_i}$ over $\F_q$ with distinct monic irreducibles $f_i$
of degrees $d_i$, a uniform $a\in\Rq$ is a unit with probability
\[
 \Pr[a\in\units{\Rq}]=\prod_i(1-q^{-d_i})
 \ \ge\ 1-\frac nq ,
\]
the bound following from the union bound over at most $n$ factors.  For a
completely split $F$ the probability is exactly $(1-q^{-1})^n$.  The next
result removes the conditioning entirely.

We use the standard chosen-plaintext experiment: the adversary submits
$m_0,m_1$, receives an encryption of $m_\beta$ for a uniform bit $\beta$, and
outputs $\beta'$, with advantage
$\Adv=\left|\Pr[\beta'=\beta]-\tfrac12\right|$.  A success probability
$\Pr[\beta'=\beta]=1$ therefore corresponds to advantage $1/2$, the maximum.

\begin{theorem}[Perfect challenge-message recovery]
\label{thm:distinguish}
Let $\sigma$ be an efficiently computable zero-preserving section, let $q$ be
prime, and let $F=X^n+1$ be squarefree over $\F_q$.  Suppose $e_{\rm pk}$,
$e_0$, and $e_1$ are sampled from $p_2$-confined laws, and that two
efficiently chosen messages $m_0,m_1$ satisfy $M(m_0)\ne M(m_1)$.  Then a
deterministic polynomial-time adversary recovers $M(m_\beta)$ for every
$a\in\Rq$, unit or not, and outputs $\beta'=\beta$ with probability one.  Its
chosen-plaintext advantage is $1/2$.
\end{theorem}

\begin{proof}
Confinement and zero preservation give
\[
 E_{\sigma,q}(0,e_{\rm pk})=E_{\sigma,q}(0,e_1)=0,
 \qquad E_{\sigma,q}(m,e_0)=M(m),
\]
so $b=as$, $c_0=bu+M(m_\beta)$, and $c_1=-au$.

Let $\widetilde a\in\F_q[X]$ be the degree-$<n$ representative of $a$, and put
\[
 d=\gcd(\widetilde a,F),\qquad h=F/d.
\]
Squarefreeness gives $\gcd(d,h)=1$.  Modulo $d$, both $a$ and $b=as$ vanish,
so $c_0\equiv M(m_\beta)$.  Modulo $h$, the class of $a$ is a unit, hence
\[
 u\equiv-a^{-1}c_1 \pmod h,
 \qquad M(m_\beta)\equiv c_0-bu \pmod h.
\]
Polynomial CRT reconstructs $M(m_\beta)$ modulo $F$ from its residues modulo
$d$ and $h$.  The degenerate cases $d=1$ and $d=F$ are included.  Comparing
the reconstructed value with $M(m_0)$ and $M(m_1)$, which differ by
hypothesis, determines $\beta$.  Polynomial gcd, extended Euclidean
computation, and polynomial CRT run in time polynomial in $n$ and $\log q$.
\end{proof}

\begin{remark}[Squarefreeness is automatic]
\label{rem:squarefree}
The hypothesis on $F=X^n+1$ is not an additional restriction in the setting of
\cite{PradhanEtAl2026}: there $n$ is a power of two and $q$ is an odd prime,
so $q\nmid n$ and $n$ is invertible in $\F_q$.  Since $F(0)=1$, the variable
$X$ does not divide $F$, and $\gcd(F,F')=\gcd(F,nX^{n-1})=1$.  Hence $F$ is
squarefree over $\F_q$ for every admissible parameter choice.
\end{remark}

\begin{corollary}[An explicit challenge pair]
\label{cor:separation}
Let $\sigma$ be the least-nonnegative or the centered section and assume
$p<q/2$.  Then $M(0)=0$ and $M(1)\ne0$ in $\Rq$, so the messages
$m_0=0$ and $m_1=1$ satisfy the separation hypothesis of
Theorem~\ref{thm:distinguish}.
\end{corollary}

\begin{proof}
The class $\overline{\CRT}(0,0)$ is zero and $\sigma$ is zero-preserving, so
$M(0)=0$.  The class $\overline{\CRT}(1,0)$ has residues $(1,0)$ modulo
$(p_1,p_2)$, hence equals $\Delta_1$ modulo $p$, and $\Delta_1\not\equiv0$
modulo $p$ because $[p_2^{-1}]_{p_1}\ne0$.  Let $c=\sigma$ applied to that
class.  The least-nonnegative section gives $c\in[0,p)$ and the centered
section gives $c\in(-p/2,p/2]$, so in both cases $0<|c|<p$, the strict lower
bound because $c\not\equiv0\pmod p$.  Since $p<q/2$ we get $0<|c|<q$, hence
$c\not\equiv0\pmod q$ and $M(1)\ne0=M(0)$.

The restriction to these two sections is needed: an arbitrary section may send
the class to $\Delta_1+jp$ for any integer $j$, and such a representative can
be divisible by $q$.
\end{proof}

\subsection{A residue-confined subgaussian error distribution}

Let $U,V$ be independent uniform bits and let $B_1=U-V$, so that $B_1$ takes
the values $-1,0,1$ with probabilities $1/4,1/2,1/4$.  Sample the coefficients
of the error independently from $p_2B_1$.

\begin{proposition}
\label{prop:law}
The law $p_2B_1$ is discrete, efficiently samplable, supported on $p_2\R$, and
bounded coefficientwise by $p_2$, hence deterministically
$\|e\|_2\le p_2\sqrt n$ for every sample.  It is subgaussian with variance
proxy
$p_2^2/2$: for every real $t$,
\[
 \mathbb E\!\left[e^{tp_2B_1}\right]\le e^{t^2p_2^2/4},
\]
whence $\Pr\left[|p_2B_1|\ge u\right]\le2e^{-u^2/p_2^2}$ for every $u>0$, and
for independent coefficients $x=(x_1,\dots,x_n)$ and every
$y\in\mathbb R^n$,
\[
 \mathbb E\!\left[e^{t\langle x,y\rangle}\right]
 \le e^{t^2p_2^2\|y\|_2^2/4}.
\]
Equivalently, in the normalization used in \cite{PradhanEtAl2026}, which
requires
$\Pr[|\langle\alpha,e\rangle|>t\sigma\|\alpha\|_2]\le2e^{-\pi t^2}$,
the law satisfies that condition with
\[
 \sigma=p_2\sqrt\pi .
\]
The law therefore satisfies the discreteness, efficient-sampling,
boundedness, and subgaussian tail conditions imposed on the error
distribution in the Setup and Section~3.2 of \cite{PradhanEtAl2026}.
\end{proposition}

\begin{proof}
Support, boundedness, discreteness, and sampling cost follow from the support
$\{-p_2,0,p_2\}$ and two fair coin flips per coefficient.  For the tail,
independence of $U$ and $V$ gives
\[
 \mathbb E\!\left[e^{tB_1}\right]
 =\frac{1+e^{t}}{2}\cdot\frac{1+e^{-t}}{2}
 =\frac{1+\cosh t}{2}
 =\cosh^2(t/2)\le e^{t^2/4},
\]
using $\cosh x\le e^{x^2/2}$.  Replacing $t$ by $tp_2$ gives the stated
moment generating function bound, that is, a subgaussian variance proxy
$s^2=p_2^2/2$.  The Chernoff bound
$\Pr[X\ge u]\le\exp(-u^2/2s^2)$ and symmetry of $B_1$ give the two-sided tail
bound.  For a linear form, independence of the coefficients multiplies the
individual bounds, producing the exponent
$t^2p_2^2\sum_jy_j^2/4$.  For the stated normalization, substituting
$u=t\sigma\|\alpha\|_2$ into the two-sided bound gives exponent
$-t^2\sigma^2/p_2^2$, which equals $-\pi t^2$ exactly when
$\sigma=p_2\sqrt\pi$.
\end{proof}

The law $p_2B_1$ is distinct from the unit-step centered-binomial family
$B_\eta$ used elsewhere in \cite{PradhanEtAl2026}.
Proposition~\ref{prop:law} therefore gives an explicit efficiently samplable
error law to which Theorems~\ref{thm:collapse} and~\ref{thm:distinguish}
apply, in a form that is checkable against the stated distributional
conditions.

The construction draws the secret, the public-key error, the encryption
multiplier, and the encryption errors from a single distribution $\chi$.
Taking that distribution to be the coefficientwise law $p_2B_1$ at every
occurrence in $\KeyGen$ and $\Enc$ gives a complete instantiation of the
family to which the two theorems apply.  

\begin{corollary}[The insecure instantiation is correct]
\label{cor:correct}
Let $\sigma$ be the least-nonnegative or the centered section, let $p<q/2$,
and let all encoder errors be $p_2$-confined.  Then every fresh ciphertext
decrypts to its message.
\end{corollary}

\begin{proof}
Collapse gives $b=as$, $c_0=bu+M(m)=asu+M(m)$ and $c_1=-au$, so the decryption
phase is exactly
\[
 c_0+c_1s=asu+M(m)-aus=M(m),
\]
with no error term at all.  Both named sections give
$\|M(m)\|_\infty<p<q/2$, so centering modulo $q$ returns the integer
polynomial $M(m)$ itself, and reduction modulo $p_1$ returns $m$.
\end{proof}

Thus the confined-law instantiation is simultaneously correct and insecure.

\begin{corollary}[Uniform asymptotic adversary]
\label{cor:asymptotic}
Let a polynomial-time parameter generator output reduced-CRT instances
satisfying the hypotheses of Theorem~\ref{thm:distinguish}, with efficiently
computable distinct message encodings and errors drawn from $p_2B_1$.  Then
the algorithm of Theorem~\ref{thm:distinguish} is a single uniform
deterministic polynomial-time chosen-plaintext adversary that succeeds with
probability one at every security parameter.
\end{corollary}

\section{Error-distribution mismatch in the RLWE transformation}
\label{sec:mismatch}

Theorem~4 of \cite{PradhanEtAl2026} transforms an ordinary RLWE sample
$(a,as+e)$ with $e\leftarrow\chi$ by scaling with $\Delta_2$, and identifies
the resulting error $\Delta_2e$ with $\CRT_{p_1,p_2}(0,e)$.  For the reduced
encoder these are different maps, and the difference is visible at the level
of distributions.  Both maps act coefficientwise, so it is enough to compare
them on a single coefficient.  Write, for an integer $x$,
\[
 \tau(x)=\Delta_2x\bmod q,
 \qquad
 \gamma(x)=\sigma\!\left(\overline{\CRT}(0,x\bmod p_2)\right)\bmod q,
\]
and let $T$ and $C$ denote the corresponding coefficientwise maps on integer
polynomials, so that $C(e)=E_{\sigma,q}(0,e)$.  The transformation of
Theorem~4 produces the error law $T_\#\chi$, while the reduced encoder
produces $C_\#\chi$.  The same-law conclusion can hold only if these two
push-forwards coincide.

For the least-nonnegative representative of the reduced CRT expression, the
discrepancy admits a closed form.  The push-forward analysis below continues
to allow an arbitrary section.

\begin{lemma}[Exact discrepancy for the least-nonnegative representative]
\label{lem:discrepancy}
Write $k=[p_1^{-1}]_{p_2}$, so that $\Delta_2=p_1k$, and let
$\gamma_0(e)=\Delta_2(e\bmod p_2)\bmod p$ be the integer returned by the
reduced CRT function on $(0,e)$.  Then, for every integer $e$,
\[
 \Delta_2e-\gamma_0(e)=p\left\lfloor\frac{ke}{p_2}\right\rfloor .
\]
Consequently the identity $\Delta_2e\equiv\CRT_{p_1,p_2}(0,e)$ in $\Rq$ holds
if and only if $\lfloor ke/p_2\rfloor\equiv0\pmod q$.  Under the no-wrap
condition $|\Delta_2e|<q/2$ used in the proof of Theorem~4, this is equivalent
to
\[
 0\le ke<p_2 .
\]
\end{lemma}

\begin{proof}
Write $e=p_2t+y$ with $y=e\bmod p_2\in\{0,\dots,p_2-1\}$.  Then
$\gamma_0(e)=\Delta_2y-p\lfloor\Delta_2y/p\rfloor$ and
$\Delta_2e=\Delta_2p_2t+\Delta_2y$, so
\[
 \Delta_2e-\gamma_0(e)=\Delta_2p_2t+p\left\lfloor\frac{\Delta_2y}{p}\right\rfloor .
\]
Now $\Delta_2p_2=p_1kp_2=pk$ and $\Delta_2y/p=ky/p_2$, whence the right-hand
side is $p(kt+\lfloor ky/p_2\rfloor)=p\lfloor k(p_2t+y)/p_2\rfloor
=p\lfloor ke/p_2\rfloor$.  Since $\gcd(p,q)=1$, this vanishes in $\Rq$ exactly
when $\lfloor ke/p_2\rfloor\equiv0\pmod q$.  Under the no-wrap condition,
\[
 \left|\frac{ke}{p_2}\right|
 =\frac{|\Delta_2e|}{p}<\frac{q}{2p}.
\]
Since $p>1$, this implies $|\lfloor ke/p_2\rfloor|<q$.  The congruence then
forces $\lfloor ke/p_2\rfloor=0$, which is equivalent to
$0\le ke<p_2$.
\end{proof}

The lemma isolates the mechanism.  The proof of Theorem~4 invokes the
small-error condition $\|\Delta_2e\|<q/2$, but that condition does not force
$\lfloor ke/p_2\rfloor$ to vanish, as the next example shows.

\begin{example}[A single integer refutes the identity]
\label{ex:single}
At the reported $p_1=65537$ and $p_2=3$ we have $k=2$ and
$\Delta_2=131074$.  Among coefficients satisfying the no-wrap premise, the
identity modulo $q$ therefore holds exactly for $0\le2e<3$, that is, for
$e\in\{0,1\}$.  Already $e=2$ fails:
\[
 \Delta_2\cdot2=262148,
 \qquad
 \CRT_{p_1,p_2}(0,2)=65537,
\]
which differ by $p=196611$, while
$\|\Delta_2\cdot2\|_\infty=262148<q/2$ for either reported ciphertext prime.
The small-error premise of Theorem~4 is therefore satisfied and its conclusion
still fails, on a single coefficient.
\end{example}

\begin{proposition}[Coefficient push-forward mismatch]
\label{prop:mismatch}
Let $q$ be prime with $p<q/2$, let $\sigma$ be any section, and let $\nu$ be a
distribution on $\Z$.  If $\nu$ charges more than $p_2$ integers that are
pairwise distinct modulo $q$, then $\tau_\#\nu\ne\gamma_\#\nu$.
\end{proposition}

\begin{proof}
The map $\gamma$ factors through $x\bmod p_2$, and there are exactly $p_2$
residue classes modulo $p_2$, so $\gamma$ takes at most $p_2$ values and
$\gamma_\#\nu$ is supported on at most $p_2$ points.

For $\tau$, write $\Delta_2=p_1k$ with $k=[p_1^{-1}]_{p_2}$.  Then $1\le
k<p_2\le p<q$ and $\gcd(p_1,q)=1$, so both factors are units modulo the prime
$q$ and $\Delta_2$ is a unit.  Hence $\tau(x)=\tau(x')$ forces $x\equiv
x'\pmod q$, so $\tau$ induces a bijection from the residues modulo $q$ charged
by $\nu$ onto the support of $\tau_\#\nu$.  By hypothesis that support has
more than $p_2$ points.  Two distributions with different support
cardinalities are distinct.
\end{proof}

\begin{corollary}[Polynomial push-forward mismatch]
\label{cor:mismatch}
Assume the hypotheses of Proposition~\ref{prop:mismatch} and let $\chi$ sample
the $n$ coefficients of $e$ independently from $\nu$.  Then $T_\#\chi\ne
C_\#\chi$.
\end{corollary}

\begin{proof}
Both $T$ and $C$ act coefficientwise, so $T_\#\chi=(\tau_\#\nu)^{\otimes n}$
and $C_\#\chi=(\gamma_\#\nu)^{\otimes n}$.  A product measure is determined by
its marginals, and the marginals differ by
Proposition~\ref{prop:mismatch}.
\end{proof}

The hypothesis is mild.  A centered binomial $B_\eta$ charges $2\eta+1$
consecutive integers, pairwise distinct modulo $q$ when $2\eta+1\le q$, so the
result applies whenever $p_2<2\eta+1\le q$.  An untruncated discrete Gaussian
assigns positive mass to every integer and therefore charges more than $p_2$
residues, since $q>p_2$.  Collapse of the kind
studied in Section~\ref{sec:attack} is the opposite extreme, in which
$\gamma_\#\nu$ degenerates to a point mass.

\begin{example}[Separation at the parameters of \cite{PradhanEtAl2026}]
\label{ex:tv}
Take $p_1=65537$ and $p_2=3$ as reported, and let $q=35175245135873$, one of
the two reported ciphertext primes, so that $\Delta_2=131074$ and
$p=196611<q/2$.  Let $\nu$ be the centered binomial $B_2$, taking
$-2,-1,0,1,2$ with probabilities $1/16,4/16,6/16,4/16,1/16$.  Since $\nu$
charges five integers that are pairwise distinct modulo $q$ and $5>3=p_2$,
Proposition~\ref{prop:mismatch} applies.  Concretely, $\gamma$ collapses
$\{-2,1\}$ onto one point and $\{-1,2\}$ onto another, so under the
least-nonnegative section
\[
 \gamma_\#\nu:\quad 0\mapsto\tfrac{6}{16},\quad
 131074\mapsto\tfrac{5}{16},\quad 65537\mapsto\tfrac{5}{16},
\]
whereas $\tau_\#\nu$ retains all five atoms with the binomial weights.  The
total variation distance between these coefficient laws is $3/8$.  Under the
centered section it is $5/8$.

The example satisfies the small-error premise used in the proof of Theorem~4,
namely $|\Delta_2e|<q/2$, and does so for every sample.  At the reported $n=8192$, every coefficient of $e$ is bounded by
$2$, so
\[
 \|\Delta_2e\|_\infty\le2\Delta_2=262148,
 \qquad
 \|\Delta_2e\|_2\le2\Delta_2\sqrt n<2.38\times10^7,
\]
against $q/2=17587622567936.5$.  The counterexample therefore lies inside the
regime the transformation assumes.
\end{example}

\begin{corollary}[Near-disjointness at the reported dimension]
\label{cor:disjoint}
Under the hypotheses of Example~\ref{ex:tv}, with $\chi=\nu^{\otimes n}$,
\[
 \TV(T_\#\chi,C_\#\chi)\ \ge\ 1-\left(\tfrac58\right)^n
\]
for the least-nonnegative section, and $1-(3/8)^n$ for the centered section.
At $n=8192$ the first bound exceeds $1-10^{-1672}$ and the second exceeds
$1-10^{-3489}$.
\end{corollary}

\begin{proof}
Let $S$ be the support of $\gamma_\#\nu$, so that $C_\#\chi$ is supported on
$S^n$.  From the table above, $\tau(x)\in S$ only for $x\in\{0,1\}$ under the
least-nonnegative section, of total mass $6/16+4/16=5/8$, which by
Lemma~\ref{lem:discrepancy} is exactly the probability that a coefficient
satisfies the identity of Theorem~4, and only for $x=0$ under the centered
section, of mass $6/16=3/8$.  By independence,
$T_\#\chi(S^n)$ is $(5/8)^n$ and $(3/8)^n$ respectively.  Taking $A$ to be the
complement of $S^n$ gives $C_\#\chi(A)=0$ and
$\TV\ge T_\#\chi(A)-C_\#\chi(A)$.  The numerical bounds follow from the exact
integer inequalities
\[
 10^{1672}5^{8192}<8^{8192},
 \qquad
 10^{3489}3^{8192}<8^{8192}.
\]
\end{proof}

At the reported dimension their total variation distance is therefore
exponentially close to one.

\begin{remark}[A smaller illustration]
\label{rem:maxsep}
The same phenomenon is visible at toy parameters.  For $p_1=3$, $p_2=2$,
$q=17$ and $\nu$ uniform on $\{-1,1\}$, both sections send the CRT class with
residues $(0,1)$ to the integer $3$, so $\gamma_\#\nu$ is the point mass at
$3$ while $\tau_\#\nu$ is uniform on $\{14,3\}\subset\Z_{17}$, at total
variation distance $1/2$.  This law has only $p_2$ points modulo $q$, so it
does not satisfy the hypothesis of Proposition~\ref{prop:mismatch}: the
condition is sufficient, not necessary.
\end{remark}

The law that the reduced encoder actually induces is the push-forward
\[
 \psi=(e\mapsto E_{\sigma,q}(0,e))_\#\chi .
\]
This identity determines $\psi$, but it does not by itself offer a reduction
from ordinary RLWE.  An ordinary RLWE sample exposes $as+e$, not $e$, so
the encoder cannot be applied to the sample.  Establishing hardness for $\psi$
requires an argument that does not presuppose access to the error term.

If the unreduced lift of Definition~\ref{def:raw} is used instead, then
$E_{\rm raw,q}(0,e)=\Delta_2e$.  Under the prime-modulus condition $\Delta_2$
is a unit modulo $q$, so scaling both components of an ordinary RLWE sample by
$\Delta_2$ is algebraically consistent with that map and preserves uniformity
of the first component.  The two encoders thus behave differently with respect
to the transformation, which is why the choice between them belongs to the
definition of the construction.

Scaling by $p_2$ gives an exact equivalence for ordinary RLWE.

\begin{theorem}[Scaling equivalence and error collapse under the reduced encoder]
\label{thm:scaling}
Let $\chi_0$ be any distribution on integer polynomials and let
$\chi=(x\mapsto p_2x)_\#\chi_0$.  Assume $p_2$ is a unit modulo $q$.
\begin{enumerate}
\item With the secret drawn uniformly from $\Rq$, ordinary RLWE with error law
$\chi$ is exactly and efficiently equivalent to ordinary RLWE with error law
$\chi_0$.
\item The same holds when a small secret is scaled alongside the error, that
is when $s=p_2s_0$ and $e=p_2e_0$ with $s_0,e_0\leftarrow\chi_0$.
\item Every sample of $\chi$ lies in $p_2\R$, so under any zero-preserving
section the reduced CRT encoder maps it to zero.  The confinement hypotheses of
Theorems~\ref{thm:collapse} and~\ref{thm:distinguish} are therefore satisfied,
and their conclusions hold whenever the remaining hypotheses of each are met.
\end{enumerate}
\end{theorem}

\begin{proof}
For the first part, let $b=as+p_2e_0$ with $s$ uniform.  The public map
$(a,b)\mapsto(p_2^{-1}a,p_2^{-1}b)$ gives second component
$(p_2^{-1}a)s+e_0$, leaves the first component uniform because $p_2$ is a
unit, and is inverted by multiplication of both components by $p_2$.  For the
second part, $b=as+e=p_2(as_0+e_0)$, so $(a,b)\mapsto(a,p_2^{-1}b)$ returns
$(a,as_0+e_0)$ exactly and leaves the first component untouched.  Both maps
are bijections on samples, carry the uniform comparison ensemble to itself,
and therefore preserve advantage in both directions.  The third part is the
computation of Theorem~\ref{thm:collapse} applied to a law supported on
$p_2\R$.
\end{proof}

The invertible scaling by $p_2$ preserves the ordinary-RLWE problem,
whereas the reduced CRT encoder erases the resulting errors.

\subsection{Relation to the conditional security theorem}
\label{sec:conditional}
Theorem~5 of \cite{PradhanEtAl2026} is conditional on decisional
CRT-RLWE hardness for the selected error law.  For $p_2B_1$ under the
reduced encoder, the CRT-RLWE error is identically zero, so that premise
is false.  Theorem~\ref{thm:distinguish} therefore exhibits an insecure
instantiation satisfying the stated distributional conditions rather
than a contradiction of the conditional implication itself.

Theorem~4 is intended to transfer ordinary-RLWE hardness to the
CRT-RLWE distribution, but Proposition~\ref{prop:mismatch} shows that
its transformation does not preserve the error law.  Thus the hardness
premise of Theorem~5 is neither automatic for the permitted
distributions nor established by Theorem~4.

\section{Implications for the security formulation}

The reduced CRT encoder and the unreduced integer lift induce
different elements of $\Rq$, by Lemma~2.3, and the difference
$pK_\sigma(m,e)$ need not vanish because $p$ is a unit modulo $q$.  A correctness or security analysis
carried out for one of the two maps therefore does not transfer to the other,
and the embedding of the CRT output into $\Rq$ is part of the mathematical
definition of the construction.

For the reduced encoder, the error distribution relevant to security in $\Rq$
is the push-forward $\psi$, not the sampled law $\chi$.  Security depends on
$\psi$, and $\psi$ can degenerate even when $\chi$ is nonconstant, bounded,
efficiently samplable, and subgaussian: the confined laws of
Section~\ref{sec:attack} give $\psi$ equal to the point mass at zero.  Excluding support contained in $p_2\R$ removes this exact collapse;
a security statement for the resulting push-forward requires an
appropriate hardness assumption.

\section{Conclusion}

Under every zero-preserving realization in $\Rq$, the reduced CRT function
erases errors supported on $p_2\R$: $e\in p_2\R$ gives $E_{\sigma,q}(0,e)=0$.
The public key then satisfies $b=as$, which yields exact secret recovery
$s=a^{-1}b$ on the unit event and, when the challenge errors are confined and
the two message encodings differ, recovery of the encoded challenge message
for every public multiplier by polynomial gcd and CRT decomposition, with
chosen-plaintext advantage $1/2$.  The coefficient law $p_2B_1$ realizes these
hypotheses while meeting the discreteness, efficient-sampling, boundedness,
and subgaussian conditions imposed on the error distribution.

For the least-nonnegative representative, the discrepancy between the reduced CRT encoder and the scaled ordinary-RLWE error is exactly \(p\lfloor ke/p_2\rfloor\), a multiple of \(p\) that need not vanish modulo \(q\). Under the no-wrap premise invoked in the proof of Theorem~4, equality modulo \(q\) holds exactly when \(0\le ke<p_2\); at the parameters of \cite{PradhanEtAl2026}, the coefficient \(e=2\) already refutes the identity while satisfying that premise. More generally, the two maps induce different error laws whenever the coefficient distribution charges more than \(p_2\) integers that are pairwise distinct modulo \(q\), including a centered binomial \(B_\eta\) whenever \(p_2<2\eta+1\le q\), and an untruncated discrete Gaussian. At the reported parameters, \(B_2\) gives coefficient-level total variation distance \(3/8\) for the least-nonnegative section and \(5/8\) for the centered section, while at the reported dimension the corresponding polynomial distributions have total variation distance exponentially close to one. Theorem~4 therefore does not establish the same-law reduction; the reduced encoder instead induces the push-forward law \(\psi\). The key-recovery result is specific to residue-confined errors, whereas the reduction mismatch applies more generally.

\backmatter

\section*{Acknowledgements}
AI tools were used in manuscript editing, development and
stress-testing of the analysis, exploratory proof checking,
verification of algebraic and numerical calculations, and drafting auxiliary
code. AI suggestions were treated as exploratory. The authors independently verified all mathematical claims, calculations, and conclusions and take full
responsibility for the article.

\begingroup
\small
\bibliography{references}
\endgroup

\end{document}